\documentclass[a4paper, 11pt, oneside, onecolumn, preprint, nonatbib]{article}

\makeatletter
\let\c@author\relax
\makeatother

\usepackage[backend=biber,hyperref=true,doi=false,url=false,isbn=false, uniquename=false, uniquelist=false, style=authoryear, natbib=true]{biblatex}
\usepackage[colorlinks=true,linkcolor=blue,citecolor=blue]{hyperref}
\usepackage{bm}
\usepackage{algorithm}
\usepackage{algpseudocode}
\usepackage{amsmath}
\usepackage{bbm}
\usepackage{makecell}
\usepackage{amsfonts}
\usepackage{amssymb}
\usepackage{amsthm}
\usepackage{fancybox}
\usepackage{multirow}
\usepackage{tikz}
\usepackage{thm-restate}
\usepackage{wrapfig}
\usepackage{mathtools}
\usepackage{booktabs}
\usepackage{tikz}
\usepackage{subcaption}
\usepackage{math}

\providecommand{\vone}{\ensuremath{\vec{1}}}

\newcommand{\vol}{\mathrm{Vol}}
\newcommand{\R}{\mathbb{R}}
\newcommand{\Sbar}{\bar{S}}
\DeclareMathOperator{\Tr}{Tr}
\DeclareMathOperator{\st}{s.t.}
\DeclareMathOperator{\Diag}{Diag}

\title{A Cheeger Inequality for Size-Specific Conductance}

\author{
Yufan Huang\\
Purdue University\\
\texttt{2019hyf@gmail.com}
\and
David F.~Gleich\\
Purdue University \\
\texttt{dgleich@purdue.edu}
}

\newcommand\blfootnote[1]{%
  \begingroup
  \renewcommand\thefootnote{}\footnote{#1}%
  \addtocounter{footnote}{-1}%
  \endgroup
}

\begin{document}
\maketitle
\begin{abstract}
  The \(\mu\)-conductance measure proposed by Lov\'asz and Simonovits is
  a size-specific conductance score that identifies the set with
  smallest conductance while disregarding those sets with volume
  smaller than a \(\mu\) fraction of the whole graph. Using
  \(\mu\)-conductance enables us to study the network structures in new
  ways. In this manuscript we study a modified spectral cut for
  $\mu$-conductance that is a natural relaxation of the integer program
  of \(\mu\)-conductance and show that the optimum of this program has a
  two-sided Cheeger inequality with \(\mu\)-conductance.
  \blfootnote{The work was funded in part by NSF CCF-1909528,
    IIS-2007481, DOE DE-SC0023162, and the IARPA Agile program. We
    thank C.~Seshadhri for discussions on this topic.}
\end{abstract}

\section{Introduction}

Graph clustering is one of the most fundamental problems in graph
analysis and has many practical applications \citep{Schaeffer_2007}. A
major class of clustering methods is based on finding small cuts,
motivated by the observation that meaningful clusters tend to be
internally well-connected while being sparsely connected to the rest
of the graph. The minimum cut problem is perhaps the most classical
example, but it often produces highly unbalanced solutions, such as
isolating a single low-degree vertex \citep{Ford_1963}. To address
this issue, the sparsest cut (or low-conductance cut) objective
instead minimizes the ratio between the size of the cut,
$\sizeof{\partial S}$, and the size of the smaller side of the partition,
$|S|$ (or, in the conductance case, the smaller volume side), and has achieved notable success
\citep{Leighton_1999,Chung_1996}.

However, a limitation of the sparsest cut and low-conductance cut
objectives is that they capture only a single structure in a
graph—the set with the worst bottleneck. In reality, the relationship
between cut size $\sizeof{\partial S}$ and set size $|S|$ can be much
more complex. \citet{Leskovec_2009} conducted a systematic study of the
so-called \emph{network community profile}, which measures the minimum
conductance over a range of different set sizes, across a large
collection of real-world graphs. They observed that, in many graphs,
the minimum-conductance sets at different size scales can have
dramatically different conductance values. For example, the minimum
conductance set of size roughly half of the graph can have conductance
several orders of magnitude larger than that of the unconstrained
minimum-conductance set.

A related size-specific cut measure is $\mu$-conductance, originally
introduced by Lov\'asz and Simonovits in the study of Markov chains
for sampling convex bodies \citep{Lovasz-1990-mixing}.
$\mu$-conductance seeks the minimum-conductance set subject to a size
constraint: the volume of the set $S$, denoted $\vol(S)$, must be at
least a $\mu$ fraction of the total graph volume.  By varying $\mu$, one
can reveal richer information about the cut structure of a graph while
ignoring sets of extremely small volume.

Despite their descriptive power, such size-specific cut profiles
cannot be computed exactly, as computing minimum-conductance sets is
NP-hard \citep{sima-schaeffer06}. A common heuristic approach is to generate many graph cuts
using different low-conductance algorithms and aggregate the results
\citep{Leskovec_2009}.  This approach faces two main
challenges. First, existing approximation algorithms for
low-conductance sets often provide guarantees only in asymptotic
Big-$O$ form, for example $O(\log n)$ or $O(\sqrt{\log n})$, which
obscures the quality of the resulting cuts in practice
\citep{Leighton_1999,Khandekar_2009,Arora_2009}. Second, algorithms
for finding low-conductance sets under size or volume constraints are
typically bicriteria, meaning they do not enforce the size constraint
exactly \citep{Andreev_2006,Arora_2009,Bansal_2014}. As a result,
existing size-specific cut profiles serve only as proxies (upper
bounds) for the true profiles. \citet{hsg23} addressed this issue by
deriving a lower bound via a relaxation of the integer program for
$\mu$-conductance and demonstrated that this bound can be computed on
large real-world graphs with millions of nodes and tens of millions of
edges.

In this paper, we show that the spectral program studied in
\citet{hsg23} satisfies a Cheeger-type inequality with
$\mu$-conductance. Cheeger’s inequality is a foundational result in
spectral graph theory, relating graph conductance to
the eigenvalues of the normalized Laplacian. The original inequality of
\citet{cheeger69} bounds the smallest Laplacian eigenvalue of a
Riemannian manifold; its discrete analogue for graphs appeared later,
with the earliest instances due to \citet{dodziuk84} and
\citet{alonmilman85}. It has played a central
role in the design of graph algorithms and has been extended in many
directions, including directed graphs and hypergraphs
\citep{Benson_2016,Lau_2023}, multiway cuts \citep{Louis_2012},
higher-order spectral gaps \citep{Kwok_2013,Lee_2014}, 
general types of diffusions \citep{Ghosh-2014-cheeger}, and conductance
across multiple graphs \citep{Koutis_2023} (which can be specialized for 
useful MinCut-specific bounds~\cite{Zhu-2016-parallel-mincut}). 

The classical Cheeger
inequality states that
\[
2 \phi_G \ge \lambda_G \ge \phi_G^2 / 2,
\]
where $\lambda_G$ is the spectral gap of the normalized Laplacian and
$\phi_G$ is the minimum conductance over all vertex sets in $G$.

Let $\lambda_{\mu}$ denote the optimum of the spectral program for
$\mu$-conductance and let $\phi_{\mu}$ denote the $\mu$-conductance.
We prove that
\[
2 \phi_{\mu} \ge \lambda_{\mu} \ge
\frac{1}{4}\left( \frac{\mu^2 \phi_{\mu} + \left(1/2 - \mu\right)\phi_G}{\mu^2 + 1/2 - \mu} \right)^2.
\]
This result can be viewed as a size-specific generalization of the
classical Cheeger inequality, relating $\mu$-conductance to the optimum
value of the corresponding spectral relaxation. When $\mu = 0$, the
size constraint is vacuous, and we recover the standard setting:
$\phi_{\mu} = \phi_G$, $\lambda_{\mu} = \lambda_G$, and the inequality
reduces to the classical Cheeger bound. In contrast, as $\mu$ approaches
$1/2$, corresponding to increasingly balanced cuts, our bound implies
\[
2 \phi_{\mu} \ge \lambda_{\mu} \ge \phi_{\mu}^2 / 4.
\]

\section{Preliminaries}

  In a weighted undirected graph \(G = (V, E, w : E \to \R)\) with
  \(n\) vertices and \(m\) edges, each edge \(uv \in E\) has weight
  \(w(u, v) \in \R\). The (weighted) degree of a vertex \(v\) is defined
  as the sum of weights of its incident edges:
  $d(v) \defeq \sum_{u \in V : uv \in E} w(v, u)$.

  For a vertex set \(S \subseteq V\), its \emph{volume} is the total degree of
  vertices in \(S\): $\vol(S) \defeq \sum_{v \in S} d(v)$. This serves as a
  size measure analogous to \(|S|\). The volume of the entire graph is
  $ \vol(G) \defeq \vol(V) = \sum_{v \in V} d(v)$. The set of edges with
  one endpoint in \(S\) and the other in \(\Sbar = V \setminus S\) is the
  \emph{cut} separating \(S\) from \(\Sbar\):
  $\partial S \defeq \{uv \in E : u \in S, v \in \Sbar \}$, with total weight
  $|\partial S| \defeq \sum_{uv \in \partial S} w(u, v)$.  Finding sparse cuts is central
  to community detection, since different communities are often weakly
  connected. However, the raw cut size ignores the relative sizes of
  \(S\) and \(\Sbar\), which can lead to highly unbalanced
  partitions. To address this, \emph{conductance} normalizes cut size
  by the smaller volume:
  \[ \phi(S) \defeq \frac{|\partial S|}{\min\{\vol(S), \vol(\Sbar)\}}. \]

  Low-conductance sets tend to induce sparse and more balanced cuts,
  since \(\vol(S) + \vol(\Sbar) = \vol(G)\). The conductance of the
  graph is then $\phi_G \defeq \min_{S \subset V} \phi(S)$, which measures the
  overall connectivity of \(G\).

  An important toolkit for studying graph structure, for instance
  conductance, is \emph{spectral graph theory}, which focuses on
  graph-associated matrices. The \emph{degree matrix} of \(G\) is the
  diagonal matrix \(\mD\) with \(\mD_{vv} = d(v)\). The \emph{graph
    Laplacian} is
\[
  \mL \defeq \sum_{uv \in E} w(u, v) \, (\vone_{u} - \vone_{v})(\vone_{u} -
  \vone_{v})^T,
\]
where \(\vone_{S}\) denotes the indicator vector of a set \(S\), and
\(\vone_{a} = \vone_{\{a\}}\). The normalized Laplacian is
$\mD^{-1/2} \mL \mD^{-1/2}$.

\section{Cheeger Inequality for $\mu$-conductance}
In this section, we first define the \(\mu\)-conductance, which provides
a more fine-grained view of a graph’s cut structure. In contrast, the
standard notion of conductance captures only a single aspect of the
graph—namely, the set corresponding to the worst bottleneck. We then
introduce the modified spectral program inspired by the
\(\mu\)-conductance integer program from \citet{hsg23}. Finally, we show
that the optimum of this spectral program satisfies a two-sided
Cheeger inequality with respect to \(\mu\)-conductance.

\subsection{The definition of \(\mu\)-conductance}

The idea of \(\mu\)-conductance is a parameterized variant of
conductance that arises from Markov chain theory
\citep{Lovasz-1990-mixing}. Basically, it computes the best set with
the smallest conductance disregarding those sets with volume smaller than
\(\mu \vol(G)\). Formally, it is defined as \footnote{Here we adopt a
  slightly different definition from the original paper. The original
  $\mu$-conductance is defined as
  \begin{align*}
  \phi_{\mu}(G) = \min_{S: \vol(S) \ge \mu, \vol(\bar{S}) \ge \mu} \sizeof{\partial S} /
  \min \setof{\vol(S) - \mu, \vol(\bar{S}) - \mu}.
  \end{align*}
  They are similar in spirit as they both neglect sets with volume
  smaller than a specific volume but the original one involves a
  perturbed conductance. The one we adopt keeps the standard
  conductance definition and only adds the size constraint.  We
  believe this provides a cleaner picture in network analysis.}
\begin{equation}
\phi_\mu(G) = \!\begin{array}[t]{l@{\,\,\,}l} \displaystyle \mathop{\text{minimize}}_{S \subset V} &
\phi(S) \\ \text{subject to} & \mu \vol(G) \! \le \! \vol(S) \! \le \! (1-\mu)\vol(G).  \end{array}
\label{eq:mu-cond}
\end{equation}
In the notation of \(\mu\)-conductance, the original conductance
$\phi_G$ can be expressed as \(\phi_0(G)\).  So by computing
\(\mu\)-conductance for multiple \(\mu\)s, we may be able to gain additional
information about a graph.  This will only be productive if $\phi_G$
arises \emph{only} from a small set in the graph. If $\phi(G)$ arises
from a set of nearly $\vol(G)/2$, then there are no differences
between conductance and $\mu$-conductance.

\subsection{A spectral program for $\mu$-conductance}

Before we introduce our modified spectral program for
\(\mu\)-conductance, let us first revisit the relationship between
conductance and the spectral cut to smooth the transition to our new
program. Basically, the problem of finding the set of smallest
conductance is equivalent to (up to a constant factor) the following
integer program
\begin{equation}
\begin{split}
\label{eq:conductance-ip}
  \displaystyle \mathop{\text{minimize}}_{S \subset V} & \quad \frac{\vpsi^T \mL \vpsi}{\vpsi^T \mD
    \vpsi} \\ \text{subject to} & \quad \vpsi = \frac{1}{\vol(S)} \vone_S - \frac{1}{\vol(\Sbar)}
  \vone_{\Sbar}
\end{split}
\end{equation}
where $\vpsi$ is a shifted and scaled indicator vector for the vertex
set $S$. The spectral cut
\begin{align}
\begin{split}
\label{eq:conductance-spectral}
  \displaystyle \mathop{\text{minimize}}_{\vx \in \R^{|V|} } & \quad \frac{\vx^T \mL \vx}{\vx^T \mD
    \vx} \\ \text{subject to} & \quad \vx^T \vd = 0.
\end{split}
\end{align}
is actually a relaxation of \eqref{eq:conductance-ip} by replacing \(\vpsi\) with vectors orthogonal
to $\vd$.

Notice that the objective in \eqref{eq:conductance-ip} is scale-invariant with regard to $\vpsi$,
it can be rewritten as
\begin{align*}
  \displaystyle \mathop{\text{minimize}}_{S \subset V} & \quad \vpsi^T \mL \vpsi
  \\ \text{subject to} & \quad \vpsi^T \mD \vpsi = 1 \\ 
                                                 &  \quad  \vpsi = \sqrt{\frac{\vol(S) \vol(\bar{S})}{\vol(G)}}\parof{\frac{1}{\vol(S)} \vone_S - \frac{1}{\vol(\Sbar)} 
                                                   \vone_{\Sbar}},
\end{align*}
where $\vpsi$ is normalized to satisfy $\vpsi^T \mD \vpsi = 1$.

We have a similar integer program for \(\mu\)-conductance, simply adding
additional volume constraints on $S$,
\begin{equation}
\begin{split}
\label{eq:mu-conductance-ip}
\displaystyle \mathop{\text{minimize}}_{S \subset V} \quad & \vpsi^T \mL \vpsi \\
\text{subject to} \quad & \vpsi = \sqrt{\frac{\vol(\Sbar)}{\vol(S)\vol(G)}} \vone_S -
\sqrt{\frac{\vol(S)}{\vol(\Sbar) \vol(G)}} \vone_{\Sbar} \\
& \vpsi^T \mD \vpsi = 1 \\
& \mu \vol(G) \le \vol(S) \le (1 - \mu) \vol(G).
\end{split}
\end{equation}
Notice that the constraint
\(\mu \vol(G) \! \le \! \vol(S) \! \le \! (1 - \mu) \vol(G) \) implies that
the entries of \(\vpsi\) must be delocalized, in other words there
will be no entries with very large magnitude and no entries with very
small magnitude. Hence this integer program can be relaxed to the
following spectral program
\begin{align}
\label{eq:mu-conductance-spectral}
\begin{split}
  \lambda_\mu = \displaystyle \mathop{\text{minimize}}_{\vx \in \R^{|V|} } \quad & \vx^T \mL \vx \\
\text{subject to} \quad &  \vx^T \vd = 0 \\
                        &  \vx^T \mD \vx = 1 \\
                        &  \|\vx\|_{\infty} \leq \sqrt{\frac{1 - \mu}{\mu \vol(G)}} \\ 
                        &  |\vx_i| \geq \sqrt{\frac{\mu}{(1 - \mu) \vol(G)}}, \quad \forall i \in V.  
\end{split}
\end{align}

\subsection{Main Result}
In this section, we demonstrate that a non-trivial relation is preserved by the relaxation from
\eqref{eq:mu-conductance-ip} to \eqref{eq:mu-conductance-spectral} and the optimum of spectral
program \eqref{eq:mu-conductance-spectral} has a two-sided Cheeger inequality with
\(\mu\)-conductance. Our main result is the following Theorem.
\begin{theorem}
\label{thm:cheeger}
Given a graph \(G\) and a constant \(0 \leq \mu \leq 1/2\), we have
\begin{align*}
  2 \phi_{\mu} \geq \lambda_{\mu} \geq \frac{1}{4} \parof{ \frac{\mu^2 \phi_{\mu} + \parof{1/2 - \mu} \phi_0}{\mu^2 + 1 / 2 - \mu}  }^2
\end{align*}
\end{theorem}

For simplicity, we assume that there exists a set \(S\) with volume
between \(\mu \vol(G)\) and \((1 - \mu)\vol(G)\); otherwise, the notion of
\(\mu\)-conductance for this value of \(\mu\) is not meaningful. The lower
bound is essentially the square of an interpolation between
\(\phi_{\mu}\) and \(\phi_0\).

At the two extreme cases, when \(\mu = 0\) we have
\(2\phi_0 \geq \lambda_0 \geq \phi_0^2 / 4\) (the standard Cheeger up to a constant
factor), and when \(\mu = 1/2\) we have
\(2\phi_{1/2} \geq \lambda_{1/2} \geq \phi_{1/2}^2 / 4\).

We point out one disadvantage of our result. Unlike the linear upper
bound \(\lambda_\mu \le 2\phi_\mu\), our bound does not yield a
square-root upper bound on \(\phi_\mu\) in terms of \(\lambda_\mu\) that
holds uniformly in \(\mu\), as the coefficient
\(\mu^2 / (\mu^2 + 1/2 - \mu)\) of \(\phi_\mu\) in the lower bound goes to
\(0\) as \(\mu \downarrow 0\).

We split the proof of Theorem~\ref{thm:cheeger} into two parts, the lower
bound of \(\lambda_\mu\) and the upper bound of \(\lambda_\mu\).

\subsection{Proof for the upper bound of \(\lambda_\mu\)}
This is the easier side of \Cref{thm:cheeger}, of which the proof was
also included in~\citet{hsg23}. We include it here for completeness.
The proof directly follows from the fact that
\eqref{eq:mu-conductance-spectral} is one relaxation of
\eqref{eq:mu-conductance-ip}.
\begin{lemma}
\label{lem:cheeger-upper-bound}
    Given a graph \(G\) and a constant \(0 \le \mu \le 1/2\), we have
    \[2 \phi_\mu \ge \lambda_\mu.\]
\end{lemma}
\begin{proof}
  Let \(S\) be the vertex set which achieves the optimal
  \(\mu\)-conductance. By our relaxation, we know
    \[
      \vpsi = \sqrt{\frac{\vol(\Sbar)}{\vol(S)\vol(G)}} \vone_S -
      \sqrt{\frac{\vol(S)}{\vol(\Sbar) \vol(G)}} \vone_{\Sbar}
    \]
    is naturally in the feasible region of
    \eqref{eq:mu-conductance-spectral}. Then, the
    corresponding objective value is
    \begin{align*}
         \vpsi^T \mL \vpsi 
    &= \frac{| \partial S| \vol(G)}{\vol(S) \vol(\Sbar)} \\
    &\leq  \frac{2 |\partial S|}{\min \{\vol(S), \vol(\Sbar) \}} \\
    &= 2 \phi_{\mu},
    \end{align*}
    which implies \(\lambda_\mu \le 2 \phi_\mu\).
\end{proof}

\subsection{Proof for lower bound of \(\lambda_\mu\)}

Now we turn to the challenging part of the proof, lower bounding
\(\lambda_\mu\) via \(\phi_\mu\). Our proof is mainly inspired by the proof in
\citet{Chung-2007-four}, with a few significant differences that we
will point out in the proof.

Let \(\vg\) denote the vector that achieves the optimum of
\eqref{eq:mu-conductance-spectral}. We order all vertices such that
\(\vg(v_1) \ge \vg(v_2) \ge \cdots \ge \vg(v_n)\) and let
$S_j = \{v_1, v_2, \ldots, v_j \}$ denote the set of first $j$ vertices for
\(j \in [n]\).  Let $h$ be the largest index such that
$\vol(S_h) \leq \vol(G) / 2$ (this implies that $h < n$).  As a result, we have
\begin{align}
  \label{eq:2}
  \min \setof{ \vol(S_i), \vol(\bar{S}_{i}) } =
  \begin{cases}
    \vol(S_i),\ \forall i \leq h \\
    \vol(\bar{S}_i),\ \forall i > h
  \end{cases}.
\end{align}
Then, we decompose $\vg$ into two vectors, denoted by $\vg_+$
and $\vg_-$ where
\begin{align*}
    \vg_+(v_i) &= \max(\vg(v_i) - \vg(v_{h+1}), 0), \\
    \vg_-(v_i) &= \max(\vg(v_{h+1}) - \vg(v_i), 0).
\end{align*}
Intuitively, $\vg_+$ and $\vg_-$ correspond to the positive and
negative components of ``shifted'' $\vg$ with $\vg(v_{h+1})$ being the
new zero point. By definition, we have
\begin{align*}
  \vg_+(v_1) \ge \ldots \geq \vg_+(v_h) \geq \vg_+(v_{h+1})  = \ldots = \vg_+(v_n) = 0
\end{align*}
and 
\begin{align*}
\vg_-(v_n) \geq  \ldots \geq \vg_-(v_{h+2}) \geq \vg_-(v_{h+1})  = \ldots = \vg_-(v_1)  = 0.
\end{align*}
For any $\vx \in \R^{V}$, let 
 \begin{align*}
   R(\vx) \defeq \frac{\vx^T \mL \vx}{\vx^T \mD \vx}.
 \end{align*}
 On one hand, because $\sum_v^{} \vg(v) d(v) = 0$,
 \begin{align*}
 \vg(v_{h+1}) \sum_v^{} \vg(v)  d(v) = 0,
 \end{align*}
 thus
 \begin{align*}
   & \quad \sum_v^{} \parof{ \vg_+(v)^2 d(v) + \vg_-(v)^2d(v)} \\
   & = \sum_v^{}\parof{ \vg(v) - \vg(v_{h+1}) }^2 d(v) \\
   &  = \sum_v^{} \vg(v)^2 d(v) - 2 \sum_v^{} \vg(v) \vg(v_{h+1}) d(v) + \sum_v^{} \vg(v_{h+1})^2 d(v) \\
     & \ge \sum_v^{} \vg(v)^2 d(v).
 \end{align*}
 On the other hand, with some casework, we have
 \begin{align*}
 \parof{\vg(u) - \vg(v)}^2 \leq \parof{\vg_+(u) - \vg_+(v)}^2 + \parof{\vg_-(u) - \vg_-(v)}^2.
 \end{align*}
 Together we have
 \begin{align}
   \lambda_{\mu} = R(\vg) &= \frac{\sum_{uv \in E}(\vg(u) - \vg(v))^2}{\sum_{v \in V} d(v) \vg(v)^2} \nonumber \\
          & \geq \frac{\sum_{uv \in E}(\vg_+(u) - \vg_+(v))^2 + \sum_{uv \in
        E}(\vg_-(u) - \vg_-(v))^2}{\sum_{v \in V} d(v) \vg_+(v)^2 + \sum_{v \in V} d(v)
      \vg_-(v)^2}. \label{eq:main-term}
 \end{align}
 In \citet{Chung-2007-four}, to get $\lambda_0 \geq \phi_0^2/2$, one can analyze
 either $R(\vg_+)$ or $R(\vg_-)$ because \eqref{eq:main-term} implies
 that $R(\vg) \ge \min \setof{R(\vg_+), R(\vg_-)}$ and
 $R(\vg_+), R(\vg_-)$ are symmetric. However, simply considering one
 of $R(\vg_+)$ and $R(\vg_-)$ can not give a clean lower bound
 involving $\phi_{\mu}$ because the entries of $\vg_+$ or $\vg_-$ can be
 arranged in a way that $R(\vg_+)$ or $R(\vg_-)$ is overly relaxed,
 which gives the trivial lower bound $\phi_0^2 / 2$.  So in our proof, we
 take $\vg_+$ and $\vg_-$ into account simultaneously. In this way the
 entries of $\vg_+$ and $\vg_-$ have less freedom and we get a lower
 bound tighter than $\phi_0^2/2$, in particular with $\phi_{\mu}$ involved.

 For convenience, we let $\vp_u$ denote $\vg_+(u)$ and $\vq_u$ denote
 $\vg_-(u)$. When there is no risk of confusion, we abuse the
 notations to let $\vp_i$ be $\vp_{v_i}$ and $\vq_i$ be $\vq_{v_i}$.
 As a result, \eqref{eq:main-term} can be written as
 \begin{align}
   \label{eq:rewritten-frac}
  \lambda_{\mu} \geq \frac{\sum_{uv \in E}^{}\parof{\vp_u - \vp_v}^2 + \sum_{uv \in E}^{} \parof{\vq_u - \vq_v}^2}
  {\sum_{v \in V}^{} d(v) \vp_u^2 + \sum_{v \in V}^{} d(v) \vq_u^2}
\end{align}

 We multiply the numerator and denominator of \eqref{eq:main-term} both by
 \begin{align*}
 \sum_{uv \in E}^{} \parof{\vp_u + \vp_v}^2 + \parof{\vq_u + \vq_v}^2 .
 \end{align*}
 By Cauchy--Schwarz inequality, the numerator is lower bounded by 
 \begin{align}
   & \quad    \sum_{uv \in E}^{} \parof{ \parof{ \vp_u - \vp_v}^2 + \parof{\vq_u - \vq_v}^2 }
     \sum_{uv \in E}^{} \parof{ \parof{ \vp_u + \vp_v}^2 + \parof{\vq_u + \vq_v}^2 }  
     \nonumber \\
   & \ge \parof{ \sum_{uv \in E}^{} \sqrt{ \parof{ \vp_u - \vp_v}^2 + \parof{\vq_u- \vq_v}^2 }
     \sqrt{ \parof{ \vp_u + \vp_v}^2 + \parof{\vq_u+ \vq_v}^2 } }^2 \nonumber \\
   & \tago{\geq}
     \parof{
     \sum_{uv \in E}^{} \sqrt{\parof{\vp_u^2 - \vp_v^2}^2 + \parof{\vq_u^2 - \vq_v^{2}}^2}
     }^2 \nonumber \\
   & \tago{\geq}
     \parof{
     \sum_{uv \in E}^{}
     \frac{\absvof{\vp_u^2 - \vp_v^2} + \absvof{\vq_u^2 - \vq_v^2}}{\sqrt{2}}
     }^2 \nonumber \\
   & \tago{=}
     \frac{1}{2}\parof{
     \sum_{i = 1}^{n-1}
     \parof{  \vp_i^2 -  \vp_{i+1}^2 } |\partial S_i| +
     \sum_{i=1}^{n-1}
     \parof{ \vq_{i+1}^2 - \vq_i^2} |\partial S_i|
     }^2 \label{eq:numerator}
 \end{align}
 where \tagr is by expanding the product and dropping nonnegative
 cross-terms, \tagr is because of AM-QM inequality, and \tagr
 is from rearranging terms.

 Similarly, the denominator is upper bounded by
 \begin{align}
  & \quad \parof{
     \sum_{v \in V} d(v) \vp_v^2 + \sum_{v \in V} d(v) \vq_v^2
     }
     \parof{
     \sum_{uv \in E}^{} \parof{ \parof{ \vp_u + \vp_v}^2 + \parof{\vq_u + \vq_v}^2 }
     } \nonumber \\
   & \tago{\leq}
     \parof{
     \sum_{v \in V} d(v) \vp_v^2 + \sum_{v \in V} d(v) \vq_v^2
     }
     \parof{
     \sum_{uv \in E}^{} 2\parof{  \vp_u^2 + \vp_v^2} + 2\parof{\vq_u^2 + \vq_v^2}
     } \nonumber \\
   & \tago{=} 2 \parof{\sum_{v \in V}^{} d(v) \parof{\vp_v^2 + \vq_v^2} }^2 \label{eq:denominator}
 \end{align}
 where \tagr is by AM-QM inequality and \tagr is by rearranging terms.

 Plugging \eqref{eq:numerator} and \eqref{eq:denominator} back into
 \eqref{eq:rewritten-frac} gives
 \begin{align}
   \lambda_{\mu} & \geq \frac{1}{4}
           \parof{
           \frac{
           \sum_{i = 1}^{n-1}
           \parof{  \vp_i^2 -  \vp_{i+1}^2 } |\partial S_i| +
           \sum_{i=1}^{n - 1}
           \parof{ \vq_{i+1}^2 - \vq_i^2} |\partial S_i|
           }
           {
           \sum_{v \in V}^{} d(v) \parof{\vp_v^2 + \vq_v^2}}
           }^2. \label{eq:1}
 \end{align}
 Now our goal is to lower bound \eqref{eq:1}.

 Before doing that, let us introduce a few indices which will
 repeatedly show up in the analysis.  Besides $h$, which is the
 largest integer such that $\vol(S_h) \leq \vol(G)/2$ as introduced
 before, we let $x, y$ be the largest indices such that
 $\vol(S_x) < \mu \vol(G), \vol(S_y) \leq (1 - \mu) \vol(G)$, respectively.
 We need these two indices because this time we cannot loosely lower
 bound every $\sizeof{\partial S_i } $ by
 $\phi_0 \min \setof{\vol(S_i), \vol(\Sbar_i)}$, we need a tighter bound
 involving $\phi_{\mu}$. Therefore we need to do some casework depending on
 $\vol(S_i)$, where $x$ and $y$ are two important cutoffs. These three
 indices are summarized in \Cref{tab:indices}. Furthermore, without
 loss of generality, we assume $\vg(v_{h+1}) < 0$; the case
 $\vg(v_{h+1}) \geq 0$ is symmetric because replacing $\vg$ with $-\vg$
 gives the same objective.

 \begin{table}[t]
    \centering
    \caption{Three key indices in our proof.}
    \begin{tabular}{rcl}
      \toprule 
      Index & Definition& Interpretation \\ 
      \midrule 
      $h$  & largest index $\st$ $\vol(S_h) \le \vol(G)/2$ & \emph{half}\\
      $x$  & largest index $\st$ $\vol(S_x) < \mu \vol(G)$ & $\mu$ fraction\\ 
      $y$  & largest index $\st$ $\vol(S_y) \leq (1 - \mu) \vol(G)$ & $(1 - \mu)$ fraction \\
      \bottomrule
    \end{tabular}
    \label{tab:indices}
  \end{table}

 With the constraints on the magnitudes of the entries and the condition
 $\vg^T \vd = 0$, intuitively there cannot be too many negative entries.
 In fact, $\vg(v_{x+1})$ is nonnegative and we have the following
 simple fact on its magnitude.
  \begin{lemma}
    \label{lem:vx+1entry}
   By definition of $x$, we have
   \[\vg(v_{x+1}) \geq \sqrt{\frac{\mu}{(1 - \mu) \vol(G)}}. \]
 \end{lemma}
 \begin{proof}
   Assume instead $\vg(v_{x+1}) < 0$, then
   \begin{align*}
     \sum_v^{} d(v) \vg(v) &= \sum_{i=1}^{x} d(v) \vg(v_i) + \sum_{i=x+1}^{n} d(v_i) \vg(v_i) \\
                        &\tago{\leq} \vol(S_{x}) \vg(v_1) + (\vol(G) - \vol(S_x)) \vg(v_{x+1}) \\
                        &\tago{\leq} \vol(S_x) \sqrt{\frac{1 - \mu}{\mu \vol(G)}} - (\vol(G) - \vol(S_x)) \sqrt{\frac{\mu}{(1 - \mu) \vol(G)}} \\
                        &\tago{<} \sqrt{\mu(1 - \mu) \vol(G)} - \sqrt{\mu(1-\mu) \vol(G)} = 0
   \end{align*}
   where \tagr is because the entries of $\vg$ are sorted in the
   decreasing order, \tagr is due to the constraints in
   \eqref{eq:mu-conductance-spectral} on the magnitude of the entries
   in $\vg$, and \tagr follows from the definition of $x$. This
   contradicts the constraint $\vg^T \vd = 0$. Therefore we know
   $\vg(v_{x+1}) \geq 0$ and by the lower bound on $\absvof{\vg_i}$ in
   \eqref{eq:mu-conductance-spectral}, we get the desired property.
 \end{proof}

 We now show that the ratio in \eqref{eq:1} is bounded below by
\begin{align}
   \label{eq:core}
        \frac{
           \sum_{i = 1}^{n-1}
           \parof{  \vp_i^2 -  \vp_{i+1}^2 } |\partial S_i| +
           \sum_{i=1}^{n - 1}
           \parof{ \vq_{i+1}^2 - \vq_{i}^2} |\partial S_i|
           }
           {
  \sum_{v \in V}^{} d(v) \parof{\vp_v^2 + \vq_v^2}}
  \geq
  \frac{\mu^2 \phi_{\mu} + (1/2 - \mu) \phi_0 }{\mu^2 - \mu + 1/2},
\end{align}
which, plugged into~\eqref{eq:1}, proves the lower bound of
$\lambda_{\mu}$ in \Cref{thm:cheeger}.
Before diving into the technical proof, let us gain some intuition
about the coefficients of $\phi_{\mu}, \phi_0$ in the bound. Consider a graph
with all vertices having roughly the same and relatively small
degrees, and a vector $\vg$ with
\begin{align*}
  \vg(v_i) \approx
  \begin{cases}
     \sqrt{\frac{1 - \mu}{\mu \vol(G)}}, \quad &1 \leq i \leq x \\
     \sqrt{\frac{\mu}{(1 - \mu) \vol(G)}}, \quad & i = x + 1 \\
    -  \sqrt{\frac{\mu}{(1 - \mu) \vol(G)}}, \quad & i > x + 1
  \end{cases}.
\end{align*}
After some slight value perturbation, $\vg$ can
satisfy $\vg^T \vd = 0, \vg^T \mD \vg = 1$ because
$\vol(S_{x}) \approx \vol(S_{x+1}) \approx \mu \vol(G)$. As a result, $\vg$ is a
feasible solution. And we have
\begin{align*}
  \vp_i \approx
  \begin{cases}
    \sqrt{\frac{1 - \mu}{\mu \vol(G)}} +
    \sqrt{\frac{\mu}{(1 - \mu) \vol(G)}}, \quad & 1 \leq i \leq x \\
    2 \sqrt{\frac{\mu}{(1 - \mu) \vol(G)}}, \quad & i = x + 1 \\
    0, \quad & i > x + 1
  \end{cases}, 
\end{align*}
and $\vq_i = 0$ for any $i$.

On one hand, for this feasible $\vg$, the numerator of the left side in
\eqref{eq:core} becomes roughly
$(\vp_x^2 - \vp_{x+1}^2) | \partial S_x| + \vp_{x+1}^2 |\partial S_{x+1}|$ and is
lower bounded by
\begin{align*}
 (\vp_x^2 - \vp_{x+1}^2) \phi_0 \vol(S_x) +  \vp_{x+1}^2 \phi_{\mu} \vol(S_{x+1})  \approx \frac{1 - 4 \mu^2}{1 - \mu} \phi_0 + \frac{4\mu^2}{1 - \mu} \phi_{\mu}.
\end{align*}
On the other hand, with some computation, we see that the denominator
is roughly $1 / (1 - \mu)$.  Overall, the lower bound for this $\vg$ is
roughly $4 \mu^2 \phi_{\mu} + (1 - 4 \mu^2) \phi_0$, and the lower bound is
tight when $\sizeof{\partial S_x} \approx \phi_0 \vol(S_x), \sizeof{\partial S_{x+1}} \approx \phi_{\mu} \vol(S_{x+1})$.

Although this does not rule out the possibility of having a tighter
lower bound with a totally different proof technique, this implies
that the lower bound we give is likely not possible to dramatically
improve within the current proof framework. In other words, applying
Cauchy--Schwarz and then considering threshold (sweep) cuts
$\sizeof{\partial S_i}$ is unlikely to give a strong lower bound.

Now, with this intuition, we establish \eqref{eq:core}. Observe that
  \begin{align}
    & \quad   \sum_{i = 1}^{n-1} \parof{  \vp_i^2 -  \vp_{i+1}^2 } |\partial S_i| \nonumber \\
    & \tago{=} \sum_{i=1}^x \parof{ \vp_i^2 - \vp_{i+1}^2} |\partial S_i|
      + \sum_{i=x+1}^h \parof{ \vp_i^2 - \vp_{i+1}^2} |\partial S_i| \nonumber\\
    & \tago{\geq} \phi_0 \sum_{i=1}^x \parof{\vp_i^2 - \vp_{i+1}^2} \vol(S_i)
      + \phi_{\mu} \sum_{i=x+1}^h \parof{ \vp_i^2 - \vp_{i+1}^2} \vol(S_i) \nonumber \\
    &\tago{=} \phi_0 \sum_{i=1}^x \vp_i^2 d(v_i) - \phi_0 \vp_{x+1}^2 \vol(S_x)
      + \phi_{\mu} \vp_{x+1}^2 \vol(S_{x}) + \phi_{\mu} \sum_{i=x+1}^{h} \vp_i^2 d(v_i) \nonumber\\
    &\tago{=} \phi_{\mu} \sum_{i=1}^h \vp_i^2 d(v_i) - (\phi_{\mu} - \phi_0) \sum_{i=1}^x (\vp_i^2 - \vp_{x+1}^2) d(v_i) \nonumber\\
    &\tago{=} \phi_{\mu} \sum_{i=1}^n \vp_i^2 d(v_i)- (\phi_{\mu} - \phi_0) \sum_{i=1}^{x+1} (\vp_i^2 - \vp_{x+1}^2) d(v_i) \label{eq:psum}
  \end{align}
  where \tagr is because $\vp_{h+1} = \ldots = \vp_n = 0$,
  \tagr uses the definitions of $\phi_0, \phi_{\mu}$ and the fact \eqref{eq:2},
  \tagr and \tagr rearrange terms, and lastly \tagr again uses
  $\vp_{h+1} = \ldots = \vp_n = 0$.

  We notice that $y + 1$ has a symmetric role with $x+1$ because $x + 1$ is the smallest integer such that $\sum_{i=1}^{x+1} d(v_i) = \vol(S_{x+1}) \geq \mu \vol(G)$,
  and $y + 1$ is the largest integer that $\sum_{i = y + 1}^n d(v_i) = \vol(\bar{S}_y) \geq \mu \vol(G)$.
  Together with the fact that $\vp$ and $\vq$ are symmetric, we get the following result
    \begin{align}
       \sum_{i = 1}^{n-1} \parof{  \vq_{i+1}^2 -  \vq_{i}^2 } |\partial S_i| 
      \geq \phi_{\mu} \sum_{i=1}^n \vq_i^2 d(v_i) - (\phi_{\mu} - \phi_0) \sum_{i=y+1}^n (\vq_i^2 - \vq_{y+1}^2) d(v_i).
\label{eq:q-sum}
    \end{align}
    Combining \eqref{eq:psum} and \eqref{eq:q-sum}, we get that
    \begin{align}
      &\quad      \frac{
        \sum_{i = 1}^{n-1}
        \parof{  \vp_i^2 -  \vp_{i+1}^2 } |\partial S_i| +
        \sum_{i=1}^{n - 1}
        \parof{ \vq_{i+1}^2 - \vq_i^2} |\partial S_i|
        }
        {
        \sum_{v \in V}^{} d(v) \parof{\vp_v^2 + \vq_v^2}} \nonumber \\
      & \ge \phi_{\mu} -
        \parof{\phi_{\mu} - \phi_0}
        \frac{\sum_{i=1}^{x+1} (\vp_i^2 - \vp_{x+1}^2) d(v_i) + \sum_{i=y+1}^n (\vq_i^2 - \vq_{y+1}^2) d(v_i)}
        {\sum_{i=1}^n d(v_i) \parof{\vp_i^2 + \vq_i^2}} \nonumber\\
      & \tago{\geq} \phi_{\mu} - \parof{\phi_{\mu} - \phi_0}\frac{\sum_{i=1}^{x+1} (\vp_i^2 - \vp_{x+1}^2) d(v_i) + \sum_{i=y+1}^n (\vq_i^2 - \vq_{y+1}^2) d(v_i)}{ \sum_{i=1}^{x+1}\vp_i^2 d(v_i) + \sum_{i=y+2}^{n} \vq_i^2 d(v_i)}\nonumber \\
      & =   \phi_{\mu} - \parof{\phi_{\mu} - \phi_0}\frac{\sum_{i=1}^{x+1} (\vp_i^2 - \vp_{x+1}^2) d(v_i) + \sum_{i=y+2}^n (\vq_i^2 - \vq_{y+1}^2) d(v_i)}{ \sum_{i=1}^{x+1}\vp_i^2 d(v_i) + \sum_{i=y+2}^{n} \vq_i^2 d(v_i)}\nonumber \\    
      & \tago{=} \phi_{0} + \parof{\phi_{\mu} - \phi_0}\frac{ \vp_{x+1}^2 \vol(S_{x+1}) +  \vq_{y+1}^2 \vol(\Sbar_{y+1})}{ \sum_{i=1}^{x+1}\vp_i^2 d(v_i) + \sum_{i=y+2}^{n} \vq_i^2 d(v_i)} . \label{eq:3}
    \end{align}
    where \tagr is because $\sum_{i=x+2}^n \vp_i^2 d(v_i) + \sum_{i=1}^{y+1} \vq_i^2 d(v_i) \geq 0$ and
    \tagr rearranges terms.
    
    As shown in \Cref{lem:vx+1entry},
    \begin{align*}
     \vg(v_{x+1}) \geq \sqrt{ \frac{\mu}{(1 - \mu) \vol(G)}}.
    \end{align*}
    We also know that
    \begin{align*}
      - \sqrt{\frac{1 - \mu}{\mu \vol(G)}} \leq  \vg(v_i) \leq  \sqrt{\frac{(1 - \mu)}{\mu \vol(G)}},
    \end{align*}
    and
    \begin{align*}
     \vg(v_{h+1}) \leq - \sqrt{\frac{\mu}{(1 - \mu) \vol(G)}}
    \end{align*}
    because we made the assumption that $\vg(v_{h+1}) < 0$.
    Therefore for any $i \leq x$,
    \begin{align*}
      \frac{\vp_i}{\vp_{x+1}} &= \frac{\vg(v_i) - \vg(v_{h+1})}{\vg(v_{x+1}) - \vg(v_{h+1})}
      = 1 +  \frac{\vg(v_i) - \vg(v_{x+1})}{\vg(v_{x+1}) - \vg(v_{h+1})} \\
      &\leq 1 + \parof{\sqrt{\frac{1 - \mu}{\mu}} - \sqrt{\frac{\mu}{1 - \mu}}} / \parof{2 \sqrt{\frac{\mu}{1 - \mu}}}
      = \frac{1}{2\mu}
    \end{align*}
    and for any $i \ge y + 2$,
    \begin{align*}
      \frac{\vq_i}{\vp_{x+1}} = \frac{\vg(v_{h+1}) - \vg(v_i)}{\vg(v_{x+1}) - \vg(v_{h+1})}
      \leq \parof{-\sqrt{\frac{\mu}{1 - \mu}} + \sqrt{\frac{1 - \mu}{\mu}}} / \parof{2 \sqrt{\frac{\mu}{1-\mu}}}
      = \frac{1}{2\mu} - 1.
    \end{align*}
Hence
    \begin{align*}
      &\quad \frac{ \vp_{x+1}^2 \vol(S_{x+1}) +  \vq_{y+2}^2 \vol(\Sbar_{y+1})}{ \sum_{i=1}^{x+1}\vp_i^2 d(v_i) + \sum_{i=y+2}^{n} \vq_i^2 d(v_i)} \\
      &\geq \frac{ \vp_{x+1}^2 \vol(S_{x+1})}{ \sum_{i=1}^{x+1}\vp_i^2 d(v_i) + \sum_{i=y+2}^{n} \vq_i^2 d(v_i)} \\
      &= \frac{  \vol(S_{x+1})}{ \sum_{i=1}^{x+1} d(v_i) \vp_i^2 / \vp_{x+1}^2  + \sum_{i=y+2}^{n} d(v_i) \vq_i^2 / \vp_{x+1}^2 } \\
      & \geq \frac{  \vol(S_{x+1})}{  \vol(S_{x+1}) / 4\mu^2 + (1 / 2\mu - 1)^2 \vol(\Sbar_{y+2})} \\
      & = \frac{  1}{ 1 / 4\mu^2  + (1 / 2\mu - 1)^2 \vol(\Sbar_{y+2}) / \vol(S_{x+1})} \\
      & \tago{\ge} \frac{1}{1 / 4\mu^2 + (1 / 2\mu - 1)^2} \\
      & = \frac{\mu^2}{\mu^2 - \mu + 1/2},
    \end{align*}
    where \tagr uses the fact that $\vol(\Sbar_{y+2}) \leq \mu \vol(G) \leq \vol(S_{x+1})$.
    Plugging the above inequality into \eqref{eq:3} establishes
    \eqref{eq:core}, and hence the lower bound of \Cref{thm:cheeger}.

\section{Computational considerations and examples}
Because of the nonconvex constraints, such as $\vx^T\mD\vx = 1$,
the program in \eqref{eq:mu-conductance-spectral} is nonconvex. In
addition, the constraints on the magnitudes of the entries make the
problem no longer a generalized eigenvalue problem. We are
not aware of any technique that can compute the global optimum of this
program.

One important use of this spectral program, as shown in \citet{hsg23},
is to lower bound the $\mu$-conductance; in other words, it corresponds
to the easier direction of \Cref{thm:cheeger},
$\phi_{\mu} \ge \lambda_{\mu}/2$. To track this lower bound effectively, we perform
a sequence of transformations as originally shown in \citet{hsg23}.

We first relax \eqref{eq:mu-conductance-spectral} to the following
semidefinite program:
\begin{align}
\label{eq:mu-conductance-sdp}
\begin{split}
  \lambda_\mu^{\text{SDP}} = \displaystyle \mathop{\text{minimize}}_{\mX \in \R^{|V| \times |V|} } \quad & \Tr \parof{\mL \mX} \\
\text{subject to} \quad &  \Tr \parof{\mD \mX} = 1\\
                        &  \Tr \parof{\vd \vd^T \mX}= 0\\
                        &  \Diag(\mX) \leq \frac{(1 - \mu)\ones}{\mu \vol(G)}\\ 
                        &  \Diag(\mX) \geq \frac{\mu\ones}{(1 - \mu)\vol(G)}.
\end{split}
\end{align}
This relaxation is obtained by lifting $\vx\vx^T$ to a
positive semidefinite matrix $\mX$, and it makes the program convex.
As a consequence, $\lambda_\mu^{\text{SDP}}/2$ is a valid lower bound
on $\phi_\mu$.

Despite being convex, semidefinite programs do not scale well and are
impractical for graphs with millions of vertices and tens of millions
of edges. However, the above SDP is weakly constrained, with only
$2|V| + 2 = 2n + 2$ constraints. By the result of Barvinok and Pataki
\citep{Barvinok_1995,Pataki_1998}, such programs admit an optimal
solution of rank $O(\sqrt{C})$, where $C$ is the number of
constraints. In our case, this implies the existence of a solution of
rank $O(\sqrt{n})$.

Motivated by this observation, we consider the classical
Burer-Monteiro
factorization \citep{BMNonlinear2003,BMLocal2005,Burer_2006}, which
factorizes the $n \times n$ decision variable $\mX$ into the product
$\mY\mY^T$ where $\mY \in \R^{n \times r}$.  The factorization automatically
eliminates the positive semidefinite constraint. The resulting
low-rank semidefinite program can be written as
\begin{align}
  \label{eq:mu-conductance-lrsdp}
  \begin{split}
      \lambda_\mu^{\text{LRSDP}} = \displaystyle \mathop{\text{minimize}}_{\mY \in \R^{n \times r} } \quad & \Tr \parof{\mL \mY \mY^T} \\
\text{subject to} \quad &  \Tr \parof{\mD \mY \mY^T} = 1 \\
                        &  \Tr \parof{\vd \vd^T \mY \mY^T}= 0 \\
                        &  \Diag(\mY \mY^T) \leq \frac{(1 - \mu) \ones}{\mu \vol(G)} \\ 
                        &  \Diag(\mY \mY^T) \geq \frac{\mu \ones}{(1 - \mu) \vol(G)}.
  \end{split}
\end{align}

Although this factorization substantially reduces the number of
variables, it renders the optimization problem nonconvex. Recent
results show that, for generic problem instances, all local optima of
low-rank semidefinite programs of this form are global optima as long
as $r = \Omega(\sqrt{C})$ \citep{BVBNonconvex2018,CifBurer2021}. In
practice, to further improve scalability, one may start with a rank
$r \ll \sqrt{n}$ and gradually increase it when the primal--dual gap
fails to improve \citep{huang2024suboptimality}. Empirically, a small
rank such as $r \approx 10$ is often sufficient to reach the global optimum
on many problem instances. Even when the solution $\mY$ is not
optimal, the dual objective still provides a valid lower bound on
$\lambda_\mu$. For more computational examples along this line, where the
spectral program and its semidefinite relaxation are used to certify
lower bounds on $\mu$-conductance on large real-world graphs, we refer
the reader to \citet{hsg23}.

\paragraph{Computational Example}

On the other hand, here we provide one simple computational example to
understand how the bound of \Cref{thm:cheeger} behaves on graphs in
practice \footnote{Code for
  this example is available at
  \url{https://github.com/luotuoqingshan/mu_cond_cheeger}.}. Because
tracking exact $\mu$-conductance is NP-hard, we pick a small graph with
only 85 vertices such that we can exactly solve the integer program of
$\mu$-conductance using \texttt{Gurobi}. To make the graph easier to
visualize, we first generate the graph in a geometric way. We randomly
sample vertices in two-dimensional space and add edges from each
vertex to its nearest neighbors.  This graph contains a cut with
conductance $\approx 0.01$ for a small set of vertices $S$.  Then to
simulate the phenomenon of many real-world graphs that
$\mu$-conductance gets much larger when $\mu$ gets closer to $1/2$, we add
random edges into the graph in an Erd\H{o}s-R\'enyi way. In particular, we
add edges between each pair of vertices outside \(S\) with probability
\(0.2\), and edges incident to \(S\) with probability \(0.001\) to
preserve the low-conductance cut. This graph ends up having 568 edges,
with an edge density $0.159$. To compute the lower and upper bounds,
we then compute $\phi_{\mu}, \phi_0$ exactly using the integer program
\eqref{eq:mu-conductance-ip}. To track $\lambda_{\mu}$, we solve the
semidefinite program relaxation \eqref{eq:mu-conductance-sdp} using
the package \texttt{CSDP}. The results are summarized in
\Cref{fig:computational-example}.

\begin{figure}[t]
  \captionsetup[subfigure]{justification=centering}
  \centering
  \subfloat{\includegraphics[width=0.5\linewidth]{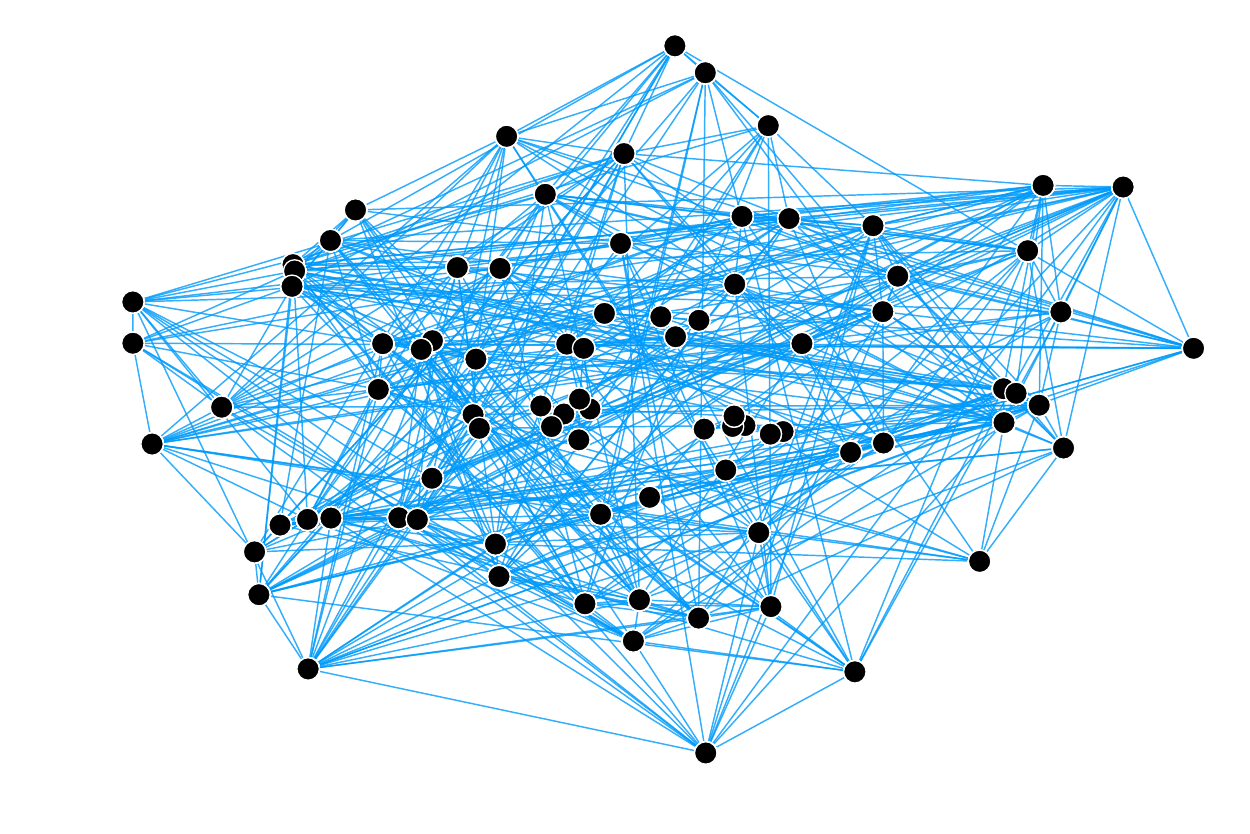}}%
  \hfill
  \subfloat{\includegraphics[width=0.5\linewidth]{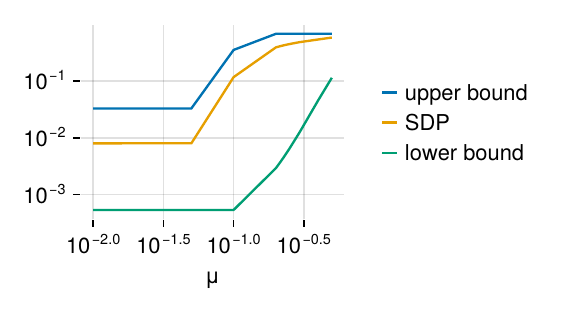}}
  \caption{An example graph with 85 vertices and 568 edges (Left) and the bounds from \Cref{lem:cheeger-upper-bound} computed on it for different $\mu$s (Right). To keep track of
    $\lambda_{\mu}$, the optimum of the spectral program, we compute its relaxation $\lambda_{\mu}^{\text{SDP}}$.
    The true $\lambda_{\mu}$ lies between the blue upper bound and the orange SDP relaxation. We see
  that the lower bound gets tighter in the large $\mu$ regime.}
  \label{fig:computational-example}
\end{figure}

\printbibliography

\end{document}